\documentclass[runningheads,orivec]{llncs}
\usepackage[T1]{fontenc}
\usepackage{graphicx}
\pdfoutput=1
\usepackage{subcaption}

\usepackage{url}
\usepackage{amssymb}
\usepackage{amsmath}
\usepackage[normalem]{ulem}

\usepackage{xspace}

\usepackage{stmaryrd}  
\usepackage{mathrsfs}  
\usepackage{algorithm2e}

\usepackage{enumitem}
\usepackage[dvipsnames]{xcolor}

\renewcommand{\top}{\mathtt{{1}}}

\renewcommand{\emptyset}{\varnothing}
\newcommand{\nodeType}[1]{\ensuremath{\mathtt{#1}}\xspace}
\newcommand{\tBAS}{\nodeType{BAS}}
\newcommand{\tOR}{\nodeType{OR}}
\newcommand{\tAND}{\nodeType{AND}}

\allowdisplaybreaks

\begin{document}
\title{Fuzzy quantitative attack tree analysis}
%
%
\author{Thi Kim Nhung Dang\inst{1}\orcidID{0000-0002-3235-5952} \and
Milan Lopuhaä-Zwakenberg\inst{1}\orcidID{0000-0001-5687-854X} \and
Mariëlle Stoelinga\inst{1,2}\orcidID{0000-0001-6793-8165}}

\authorrunning{Dang et al.}
%
\institute{University of Twente, Enschede, the Netherlands  \\
\email{\{t.k.n.dang, m.a.lopuhaa, m.i.a.stoelinga\}@utwente.nl}\\
\and Radboud University, Nijmegen, the Netherlands \\
\email{m.stoelinga@cs.ru.nl}\\
}
\maketitle              
\begin{abstract}

Attack trees are important for security, as they help to identify weaknesses and vulnerabilities in a system.  Quantitative attack tree analysis supports a number security metrics, which formulate important KPIs such as the shortest, most likely and cheapest attacks.

A key bottleneck in quantitative analysis is that the values are usually not known exactly, due to insufficient data and/or lack of knowledge. Fuzzy logic is a prominent framework to handle such uncertain values, with applications in numerous domains. While several studies proposed fuzzy approaches to attack tree analysis, none of them provided a firm definition of fuzzy metric values or generic algorithms for computation of fuzzy metrics. 

In this work, we define a generic formulation for fuzzy metric values that applies to most quantitative metrics. The resulting metric value is a fuzzy number obtained by following Zadeh's extension principle, obtained when we equip the basis attack steps, i.e., the leaves of the attack trees, with fuzzy numbers. In addition,  we prove a modular decomposition theorem that yields a bottom-up algorithm to efficiently calculate the top fuzzy metric value. 

\keywords{Attack trees  \and quantitative analysis \and fuzzy numbers.}
\end{abstract}
\section{Introduction}\label{sec:intro}
\paragraph{Attack trees.} 

Attack trees (ATs) \cite{schneier1999modeling} are a popular tool for modeling and analyzing security risks. They provide a structural way to identify vulnerabilities in a system, by decomposing the attacker's goal into subgoals, down to basic attack steps that a malicious actor can take to reach said objective. An attack tree consists of basic attack steps (BASs) representing atomic adversary actions, and intermediate AND/OR-gates whose activation depends on the activation of their children. The attacker's goal is to activate the root (top node), see Fig.~\ref{fig:ex_SAT} for an example. ATs can be trees or directed acyclic graphs (DAGs). ATs have been supported by commercial tools~\cite{isograph,risktree,SecurITree} and equipped with semantics~\cite{mauw2006foundations,kumar2015quantitative}.

\begin{figure}[!t]
     \centering
     \includegraphics[width=0.45\textwidth]{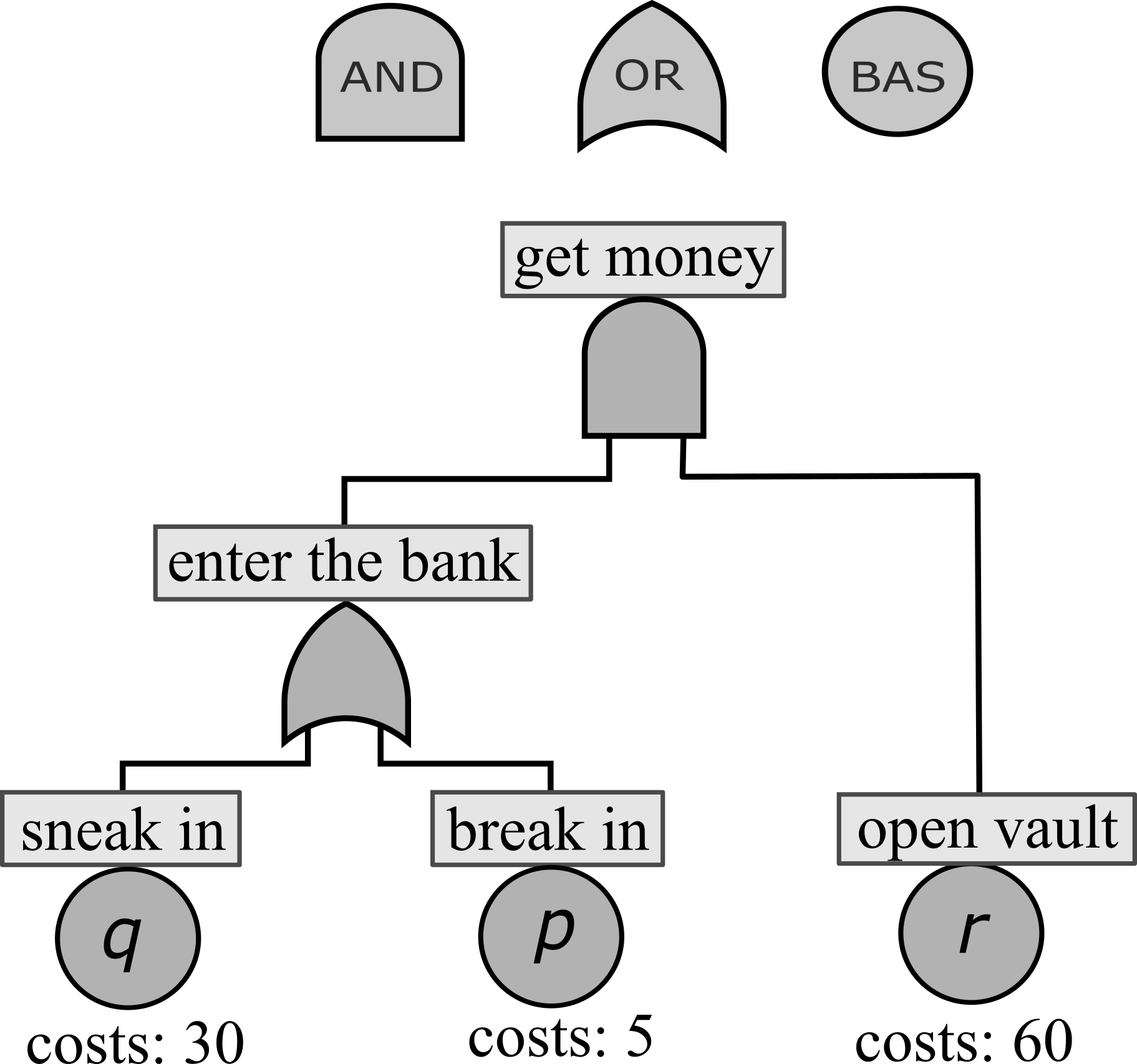}
     \caption{The AT model visualises the attack steps by which an attacker can illegally take money from a bank.  The attacker needs to enter the bank by breaking in or sneaking in, and also needs to open a vault. Sneaking in, breaking in, and opening a vault cost 30, 5 and 60 minutes, respectively. Hence, the quantitative metric minimal cost for the attacks is $\min(30+60,5+60)=65$.}
     \label{fig:ex_SAT}
\end{figure}

\paragraph{Quantitative analysis.} Beyond qualitative analysis, ATs are also used to calculate important security metrics of the system, e.g., the minimal cost (in money, time or resources) the attacker needs to spend for a succesful attack, or the probability of a succesful attack. Such metrics are obtained by assigning an attribute value to each BAS, such as the cost needed to perform that BAS, and using this as input to calculate the security metric. When the AT is treeshaped, the metric is quickly calculated using a bottom-up algorithm, propagating values from the BASs to the top. For DAG-shaped ATs this problem is NP-complete, but good heuristics exist \cite{lopuhaa2022efficient}. These algorithms are formulated in the generic algebraic structure of semirings, allowing them to be employed to a vast range of security metrics including cost, time, skill, damage, etc.

\paragraph{Uncertain parameters.} The methods described above assume that all BAS parameters are known exactly. However, this is problematic in practice: statistics on attacker capabilities may be hard to obtain, and because of the fast-changing nature of the field historical data are only of limited use. Obtaining accurate and realistic parameter values is a key bottleneck in quantitative security analysis. In its absence, there is a great need for methods that allow us to deal with uncertain and approximately known parameter values.

\paragraph{Fuzzy theory.} Fuzzy theory is a prominent framework in which parameter uncertainty and its effect on a calculation's outcome can be expressed mathematically. It has been successfully used in many applications, including machine learning \cite{couso2019fuzzy}, reliability engineering \cite{bowles1995application}, and computational linguistics \cite{massanet2014new}. Rather than exact (`crisp') values, e.g., $x = 3$, each parameter is assigned a range of values, and to each of these a possibility value in $[0,1]$ is assigned by means of a \emph{membership function}. Often, only functions of a specific form are considered, leading to the definition of triangular, trapezoidal, etc. fuzzy numbers \cite{Jezewski2017theory}.

While fuzzy theory has been applied to AT analysis before \cite{Komal2023FATA,wang2021cybersecurity,Li2020vehicle,Garg2014anAT,wen2021risk}, much of the earlier work lacks mathematical rigor, and none of these apply fuzzy theory to quantitative analysis. As a result, there are no algorithms for calculating AT metrics with fuzzy parameters. In fact, to our knowledge the fuzzy counterpart of quantitative AT analysis has not been defined yet. A key technical hurdle is that the operations typically used in AT analysis do not preserve popular fuzzy number types: for instance, the OR-gate corresponds to the operation min for the minimal cost metric, and applying min to two triangular fuzzy numbers does not yield a triangular fuzzy number.

\paragraph{Contributions.} Our first contribution is a clear, mathematically rigorous definition of fuzzy AT metrics. Because these are defined for general fuzzy numbers, rather than specific subtypes such as triangular fuzzy numbers, we sidestep the problem that these subtypes are not preserved under AT metric operations; instead, our definition works for the generic semiring framework defined in \cite{lopuhaa2022efficient}. We show that our definition naturally follows from Zadeh's extension principle \cite{zadeh1965fuzzy}, a general approach for extending functions to fuzzy numbers.

Having defined
fuzzy AT metrics, we furthermore develop a linear-time, bottom-up algorithm for calculating them for tree-shaped ATs. We show the validity of this algorithm by showing that fuzzy AT metrics are susceptible to \emph{modular analysis}: when an AT has a module, i.e., a minimally connected subcomponent, a fuzzy metric can be computed by first calculating the metric for the module and then for its complement. When an AT has many modules, this substantially speeds up computation. When an AT is tree-shaped, every node is a module, proving the validity of the algorithm.

Our algorithm generalizes the bottom-up algorithm for crisp AT metrics from \cite{lopuhaa2022efficient}. Unfortunately, the algorithm for DAG-shaped metrics from that paper does not transfer to the fuzzy setting. The key reason is that fuzzy numbers do not form a semiring, as we show in this paper. Fuzzy metrics for DAG-shaped ATs require a radically new approach, and we leave this for future work.



Summarized our contributions are:
\begin{enumerate}[label=\textbf{\arabic*.},leftmargin=1.6em]
\item	A rigorous, general definition of fuzzy AT metrics;
\item	A bottom-up algorithm for computing fuzzy metrics in tree-structured ATs;
\item   A proof of modular decomposition for fuzzy AT metrics.
\end{enumerate}



\section{Related work}\label{sec:related_work}

Below, we provide a literature review for computation of metrics with fuzzy numbers applied to attack trees and the related formalism of fault trees.

\paragraph{Attack tree analysis with fuzzy numbers.}
An intuitionistic fuzzy set was used to represent the uncertainty and hesitancy present in data~\cite{Komal2023FATA}, or an attack-defense model was proposed~\cite{wang2021cybersecurity,Garg2014anAT}, or using a fuzzy analytic hierarchy process to establish a successful probability model of cyber attack~\cite{wen2021risk,Li2020vehicle}. However, there have been several studies on the approach of involving fuzzy attribution in fault tree analysis (FTA) summarized~\cite{yazdi2023FTAimprovements,kabir2018review,ruijters2015FTA,kabir2017anoverview,mahmood2013FFTA} for many years.

\paragraph{Fault tree analysis.}
Fault trees can be considered as the safety variant of attack trees: whereas attack trees indicate how malicious attacks propagate through a system and lead to damage, fault trees indicate how unintended failures propagate and lead to system level failures. Therefore, leaves of a fault tree model component failures and are called basic events (BEs). 
Due to their similarities, many approaches to fuzzy fault tree analysis can also be applied to attack trees. Comprehensive literature surveys on fault trees with fuzzy numbers can be found in~\cite{yazdi2023FTAimprovements,mahmood2013FFTA,ruijters2015FTA,kabir2017anoverview}.

\paragraph{Fault tree analysis with fuzzy probabilities.}
Fuzzy set theory was firstly used in fault tree analysis by Tanaka et al.~\cite{tanaka1983fault} to address the problem of uncertain BEs failure. In the paper, Zadeh's extension principle was used to estimate the possibility of system failure. The failure possibility of the basic events and top event were represented as trapezoidal fuzzy numbers. 

Singer~\cite{singer1990fuzzysetFT} considered the distribution of BEs as fuzzy numbers. The membership function is continuous and is approximated by left and right functions called L-R type fuzzy numbers~\cite{Dubois1979FuzzyRA}. Here, L-R type fuzzy numbers are defined by a triplet $(m,a,b)$, where $m,a,b$ are positive real numbers. The author extended algebraic operations on the triplet of L-R type fuzzy numbers and calculated the possibility distribution of the system.

Kim et al.~\cite{kim1996multilevel} evaluated the possibility of system failure. Similar to~\cite{singer1990fuzzysetFT}, L-R type fuzzy numbers are used as the possibilities of BEs. The value $m$ of the triplet $(m,a,b)$ is evaluated by four-expert valuations in the form of triangular fuzzy numbers (TFNs). Each value $m$ is determined to calculate the optimistic and pessimistic possibilities of a system accident. Finally, two cases of possibilities - the pessimistic possibility of system failure with major TFN and the optimistic one with minor TFN - were determined.

Lin et al.~\cite{lin1997hybridFTA} estimated failure possibility of ambiguous events. For this purpose, the linguistic variables describing the evaluation data are expressed in triangular or trapezoidal fuzzy numbers denoting failure possibilities. The fuzzy possibility of a top event is calculated using the $\alpha$-cut fuzzy operators.

Peng et al.~\cite{peng2008approach} presented an approach to fault diagnosis of communication control systems. All probability values of the fault tree were converted to uniform triangle fuzzy numbers. The fuzzy probability of the top event was then calculated using Zadeh's principle. A fault tree (FT) consisting of only OR-gates was shown as an analytical example to determine the confidence interval of probability of top event and achieve fuzzy reasoning diagnosis result.






\paragraph{Fault tree reliability analysis with interval arithmetic.}

Purba et al.\cite{Purba2015fuzzyprob} developed a fuzzy probability based fault tree analysis to propagate and quantify epistemic uncertainty raised in basic events. BE reliability characteristics are described in fuzzy probabilities. From the BE fuzzy probabilities, the matrix of fuzzy probabilities of the minimal cut sets is generated and then the top event fuzzy probability is quantified using the Fuzzy multiplication rule in engineering applications.

Purba et al.~\cite{Purba2022FuzzyPA} proposed a fuzzy probability and $\alpha$‐cut based‐FTA approach. Each fuzzy probability distribution of BEs is represented uniquely by an $\alpha$-cut. The top event $\alpha$-cut is quantified into the best estimate $\alpha$-cut, the lower bound $\alpha$-cut, and the upper bound $\alpha$-cut follow fuzzy arithmetic operations on $\alpha$-cuts of BEs. The approach was verified by evaluating the reliability of a complex engineering system and the results are compared to the reliability of the same system quantified by conventional FTA.

\paragraph{Fuzzy FTA by conversion of fuzzy number of BEs to crisp probability of BEs.}

Hu et al.~\cite{Hu2019FFTA} developed an FFTA methodology for analyzing above-ground walled storage system failures. Expert elicitation and fuzzy logic was used to manipulate the ambiguities and vagueness in the linguistic variables of BEs. Fuzzy probability BE was defuzzified to a crisp number. The resultant crisp probability of BEs were used as inputs to generate crisp probability of the top event.

At the time of this writing, fuzzy analysis has not been studied for ATs. The literature has introduced fuzzy analysis of FTs, but it only addresses certain types of fuzzy numbers (trapezoidal, triangular, etc.). This paper thus provides a general mathematical framework for fuzzy analysis of ATs.







\section{Fundamentals of fuzzy theory}\label{sec:fundamentals_fuzzy}

Fuzzy set theory was introduced by L.A. Zadeh~\cite{zadeh1965fuzzy} to deal with problems in which vagueness is present. Instead of considering elements $x$ of a set $X$ with a fixed value, we consider fuzzy elements $\mathsf{x}$ which can have a range of possible values; the extent to which $\mathsf{x}$ can be equal to $x$ is expressed by the \emph{membership degree} of $x$ in $\mathsf{x}$, which is a value $\mathsf{x}[x] \in [0,1]$. The value $\mathsf{x}[x]$ is the confidence one has that $\mathsf{x}$ has value $x$. Here $\mathsf{x}[x] = 1$ denotes full membership, while $\mathsf{x}[x] = 0$ denotes no membership.

For instance, the time needed to perform an attack may be given as a real number, e.g. $x = 3 \in \mathbb{R}$; but often the exact time needed is not known precisely, and can be somewhere around $3$. This can be represented by a fuzzy number $\mathsf{x}\colon \mathbb{R} \rightarrow 1$ which is $0$ everywhere except close to $3$, and which has a maximum at $3$ (see Fig.~\ref{fig:non-fuzzy_vs_fuzzy}). 

\begin{figure}[!t]
     \centering
     \begin{subfigure}[b]{0.3\textwidth}
         \centering
         \includegraphics[width=\textwidth]{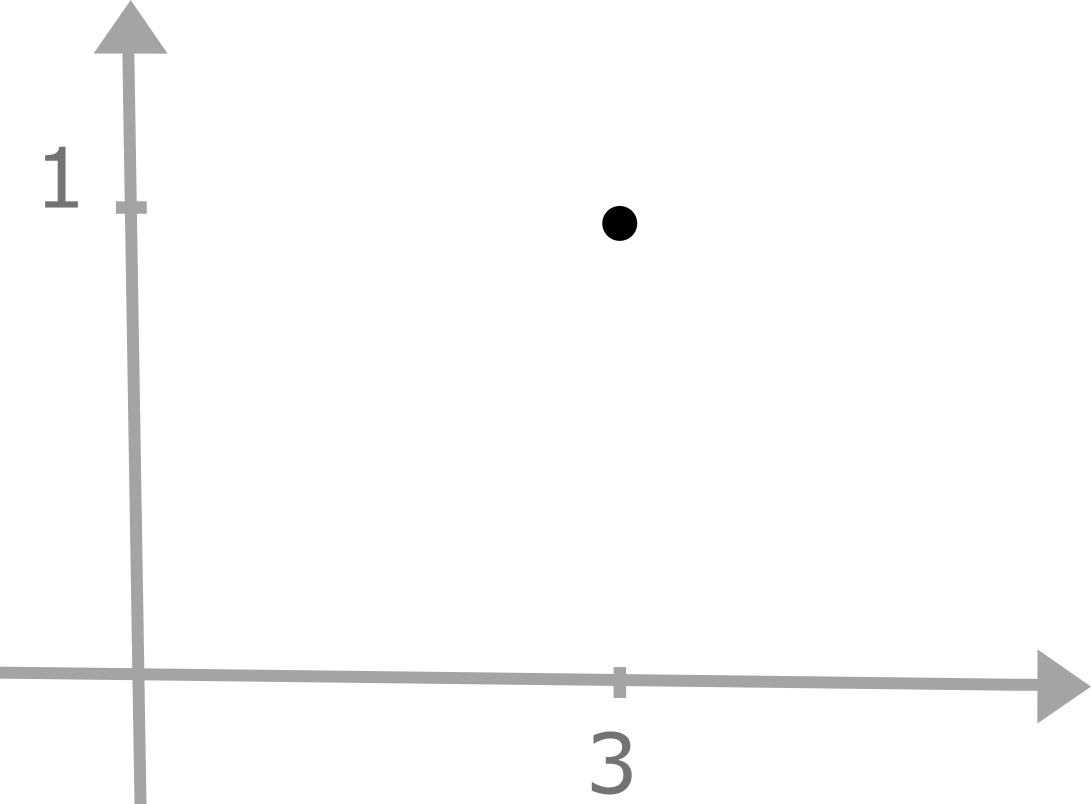}
         \caption{}
         \label{fig:classic_element}
     \end{subfigure}
     \hspace{5em}
     \begin{subfigure}[b]{0.3\textwidth}
        \hfill
        \centering         
        \includegraphics[width=\textwidth]{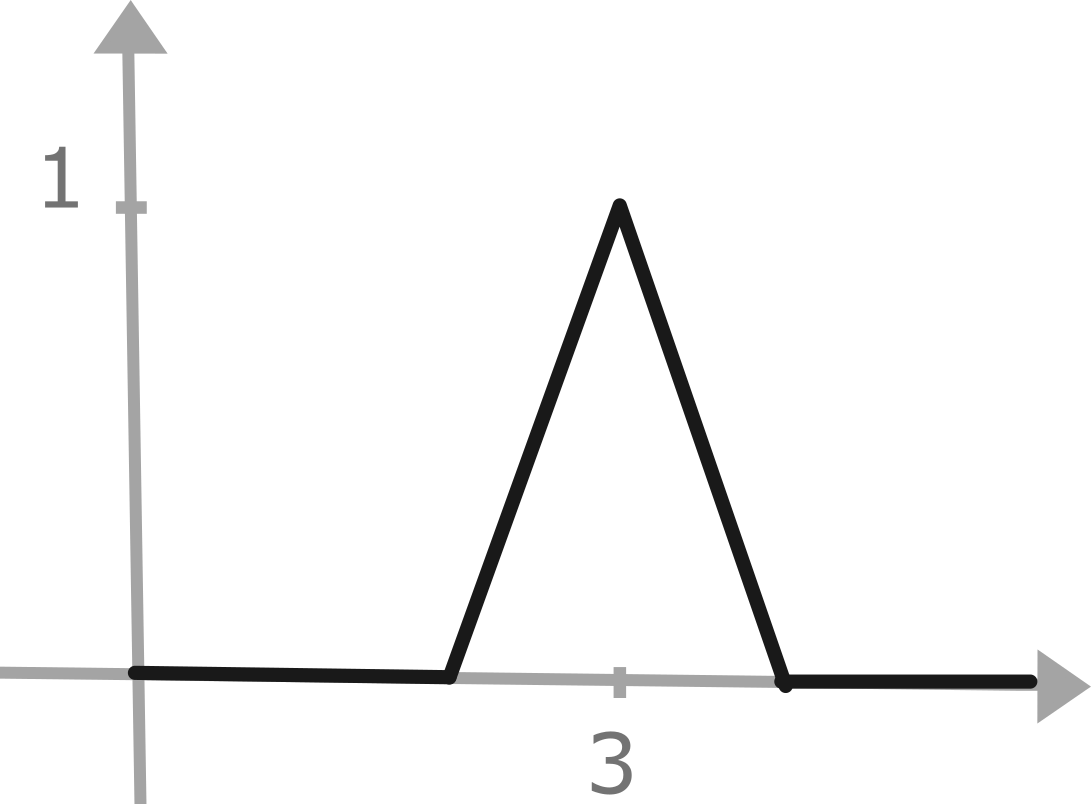}
        \caption{}
        \label{fig:fuzzy_element}
     \end{subfigure}
     \caption{A non-fuzzy, `crisp' element $x$ (a) and a fuzzy element $\mathsf{x}$ (b).}
     \label{fig:non-fuzzy_vs_fuzzy}
\end{figure}

\begin{definition}
Let $X$ be a set. A \emph{fuzzy element} of $X$ is a function $\mathsf{x}\colon X \rightarrow [0,1]$. The set of all fuzzy elements of $X$ is denoted $\mathbf{F}(X) := \{ \mathsf{x} \ | \ \mathsf{x} \colon X \rightarrow [0,1] \}$.
\end{definition}

In the literature, fuzzy elements are usually called \emph{fuzzy sets} \cite{zadeh1965fuzzy}, on the basis that the membership function $\mathsf{x}\colon X \rightarrow [0,1]$ generalizes the indicator function $1_S\colon X \rightarrow \{0,1\}$ of a set $S \subseteq X$; thus a fuzzy set can be thought of as a set of which elements can have partial membership. Instead, we use the term \emph{fuzzy element} to stress that in this paper, fuzzy elements are used to express the uncertainty of individual values, as in Fig.~\ref{fig:fuzzy_element}, rather than the uncertainty of set membership. A fuzzy element $\mathsf{x}$ behaves similarly to a probability density function in that the uncertainty of an element of $X$ is expressed by a function on $X$.

Our definition of fuzzy element is very general. Many works in the literature restrict the form of the function $\mathsf{x}\colon X \rightarrow [0,1]$ to make computation more convenient, especially for $X = \mathbb{R}$, i.e., for so-called \emph{fuzzy numbers}. Thus there exist triangular, trapezoidal, Gaussian, etc. fuzzy numbers \cite{Jezewski2017theory,Czogala2000fuzzy}.

\begin{example}
Consider real numbers $a \leq b \leq c \leq d$. The \emph{trapezoidal fuzzy number} $\mathsf{trap}_{a,b,c,d} \in \mathbf{F}(\mathbb{R})$  is defined as (see Fig. \ref{fig:ex:trapezoidal}):
\begin{align}
    \mathsf{trap}_{a,b,c,d}[x] &= 
    \begin{cases}
        \tfrac{x-a}{b-a}, & \textrm{if } a < x \leq b,\\
       1, & \textrm{if } b<x<c,\\
        \tfrac{d-x}{d-c}, & \textrm{if } c \leq x < d,\\
        0, & \text{otherwise}.
    \end{cases}
\end{align}

The trapezoidal fuzzy number $\mathsf{trap}_{a,b,c,d}$ has the maximal membership degree of $1$, i.e., $\mathsf{trap}_{a,b,c,d}[x]=1$ for all $x\in [b,c]$. At the same time, $a$ and $d$ are the lower and upper
bounds of its support, respectively. In case $b=c$, we have a \emph{triangular fuzzy number} $\mathsf{tri}_{a,b,d}$.
\end{example}

\begin{figure}[ht]
    
    \begin{center}
    \includegraphics[width=0.5\textwidth]{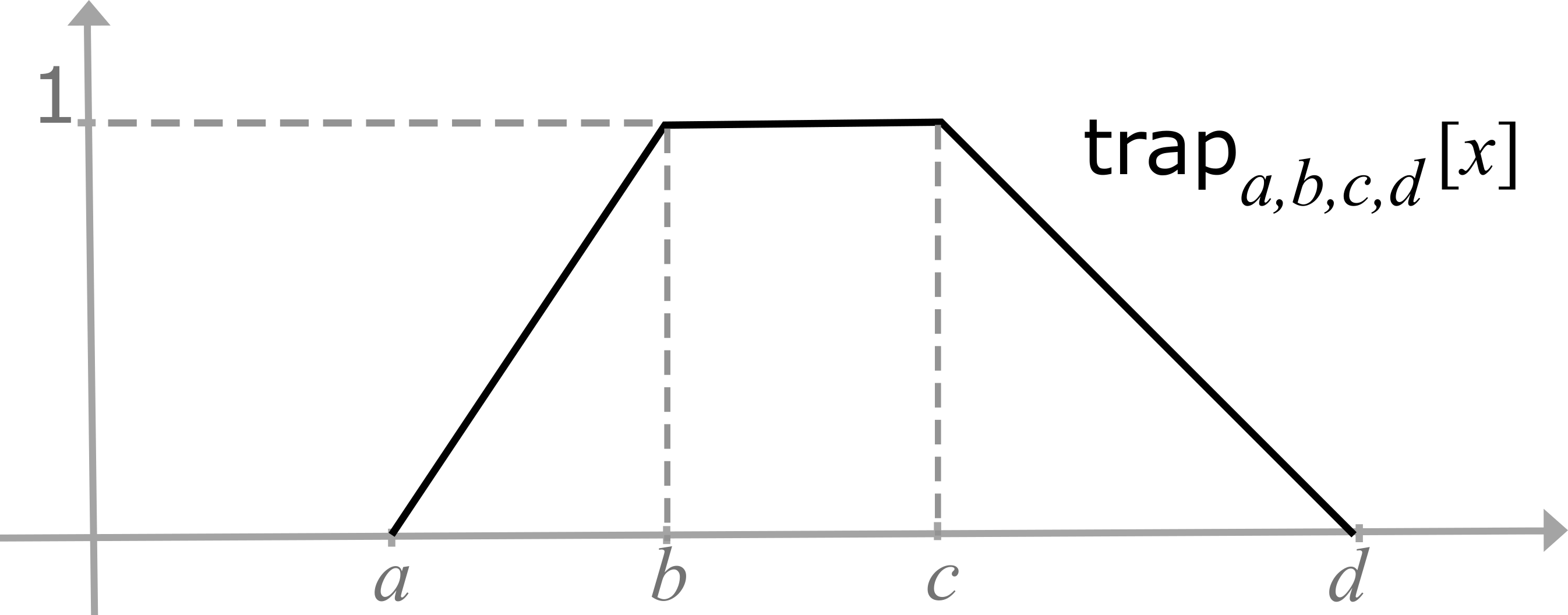}
    \end{center}
    \caption{The trapezoidal fuzzy number $\mathsf{trap}_{a,b,c,d}$.}
    \label{fig:ex:trapezoidal}
\end{figure}

For notational convenience we occasionally abbreviate $\mathsf{x}$ via a list of membership values $x \mapsto \mathsf{x}[x]$, omitting $x$ for which $\mathsf{x}[x] = 0$. For example, $\mathsf{x} = \{1 \mapsto 0.7,2 \mapsto 0.5\} \in \mathbf{F}(\mathbb{Z})$ is defined by
\begin{align*}
    \mathsf{x}[x] = 
    \begin{cases}
    0.7, & \textrm{ if $x=1$},\\
    0.5, & \textrm{ if $x=2$},\\
    0, & \textrm{ otherwise}.
    \end{cases}
\end{align*}

Arithmetic operations on fuzzy elements are performed following Zadeh's extension principle~\cite{Jezewski2017theory,deBarros2017afirst,zadeh1975conceptIII,zadeh1975conceptII,Zadeh1975conceptI,zadeh1965fuzzy}. This principle provides a framework to apply functions and arithmetic operations on sets to their fuzzy elements. Before giving the full definition, we motivate it by an example.

\begin{example}\label{ex:fuzzy_add}
    Consider $\mathsf{x},\mathsf{y} \in \mathbf{F}(\mathbb{N})$ given by
\begin{align*}
    \mathsf{x} &= \{\ 2 \mapsto 0.4,\ 3 \mapsto 1 \},\\
    \mathsf{y} &= \{\ 5 \mapsto 1,\ 6 \mapsto 0.6 \}.
\end{align*}

We wish to calculate the addition of $\mathsf{x}$ and $\mathsf{y}$, which we write as $\mathsf{x}\widetilde{+}\mathsf{y}$. This is also an element of $\mathbf{F}(\mathbb{N})$ and so we must specify the confidence $(\mathsf{x}\widetilde{+}\mathsf{y})[z]$ that the sum values to $z$, for all $z \in \mathbb{N}$. Consider $z = 8$; the sum values to $8$ only in one of these two cases:
\begin{itemize}
\item $\mathsf{x}$ values to 2 and $\mathsf{y}$ values to 6;
\item $\mathsf{x}$ values to 3 and $\mathsf{y}$ values to 5.
\end{itemize}
Our confidence that $\mathsf{x}$ values to 2 is $\mathsf{x}[2] = 0.4$, and our confidence that $\mathsf{y}$ values to 6 is $\mathsf{y}[6] = 0.6$. Our confidence that both of these are true, i.e., that the first case holds, is then $\min\{0.4,0.6\} = 0.4$. Similarly, our confidence that the second case holds is $\min\{1,1\} = 1$. Our confidence $(\mathsf{x}\widetilde{+}\mathsf{y})[8]$ that the sum values to 8 is then the confidence that either of the two cases above holds; this is expressed by the maximum, so 
\[
(\mathsf{x}\widetilde{+}\mathsf{y})[8] = \max\{0.4,1\} = 1.
\]
Similarly one can calculate $(\mathsf{x}\widetilde{+}\mathsf{y})[z]$ for other values of $z$, by taking all possible outcomes of the sum and calculating their confidence. This yields
\[
\mathsf{x} \widetilde{+} \mathsf{y} = \{7 \mapsto 0.4, 8 \mapsto 1, 9 \mapsto 0.6\}.
\]




\end{example}


The idea behind Example \ref{ex:fuzzy_add} can be applied to general multivariate functions. The only change that needs to be made is that in general, there may be infinitely many pairs $(x,y)$ such that $f(x,y) = z$; therefore one needs to take the supremum over all $\min\{\mathsf{x}[x],\mathsf{y}[y]\}$ rather than the maximum.

\begin{definition}[Zadeh's Extension Principle]\label{def:extension_principle}
Let $f$ be a multiargument function $f:X_1\times X_2\times\dots \times X_n \to Y$. The \emph{Zadeh extension} of $f$ is the function $\tilde{f}\colon \mathbf{F}(X_1) \times \ldots \times \mathbf{F}(X_n) \rightarrow \mathbf{F}(Y)$ defined as:
\[
\tilde{f}(\mathsf{x}_1,\ldots,\mathsf{x}_n)[y] = \begin{cases}
    \underset{\substack{(x_1,x_2,\dots,x_n) \in \prod_i X_i\colon\\
    f(x_1,x_2,\dots,x_n)=y}}{\sup}\min\limits_{i=1,\dots,n}\ \mathsf{x}_i[x_i], & f^{-1}(y) \neq \emptyset,\\
    0 & f^{-1}(y) = \emptyset.
\end{cases}
\]
\end{definition}

Based on the extension principle, different arithmetic operations on fuzzy numbers have been defined~\cite{Basiura2015advances,tanaka1983fault,deBarros2017afirst,liang1991fuzzyFT,peng2008approach}. As a result of Definition~\ref{def:extension_principle}, addition and subtraction operations on fuzzy numbers typically have straightforward formulations. E.g., for two trapezoidal fuzzy numbers we have
\begin{align*}
    \mathsf{trap}_{a_1,a_2,a_3,a_4} \  \widetilde{+} \ \mathsf{trap}_{b_1,b_2,b_3,b_4} &= \mathsf{trap}_{a_1 + b_1, a_2 + b_2, a_3 + b_3, a_4 + b_4},\\
    \mathsf{trap}_{a_1,a_2,a_3,a_4} \ \widetilde{-} \ \mathsf{trap}_{b_1,b_2,b_3,b_4} &= \mathsf{trap}_{a_1 - b_4, a_2 - b_3, a_3 - b_2, a_4 - b_1}.
\end{align*}
Multiplication and division, however, are nonlinear operations that produce fuzzy numbers of different types than the operands; for example, the quotient of two trapezoidal fuzzy numbers is itself not trapezoidal. For convenience and to simplify the computation, the resulting fuzzy number
can be approximated by a fuzzy number of the same type. The computation and visualisation of these estimations can be found in~\cite{Basiura2015advances}.

In section \ref{sec:fuzzy_metrics}, we will apply the general fuzzy element framework to formulate fuzzy attack tree metrics. Unfortunately, the operators considered in AT analysis, such as $\min$, do not preserve triangular, trapezoidal, etc. fuzzy numbers. We therefore need to work with fuzzy numbers and Zadeh extensions in full generality as defined above.


\section{Attack trees}\label{sec:AT}
In this section, we provide a brief overview of ATs as presented in \cite{lopuhaa2022efficient}. Attack trees are hierarchical graphical models that illustrate the attack process. The trees are usually drawn inverted, with the root node located at the top of the tree and branches descending from the root to the lowest levels of the tree -- the leaves.  The root node represents the attacker's overall objective. The leaves in ATs are called \emph{Basic Attack Steps} (BASs) representing the attacker's activities.  Nodes between the leaves and the root node depict transitional states or attacker sub-goals. These intermediate steps are equipped with \emph{logical gates} that indicate whether an intermediate step succeeds, e.g. the AND-gate succeeds if all input children succeed, the OR-gate is successful if at least one child does succeed.

\begin{definition}\label{def:attack_tree}\emph{\cite{lopuhaa2022efficient}}
    An \emph{attack tree} is a tuple $T=(N,E,t)$, where $(N,E)$ is a rooted directed acyclic graph, and $t$ is a map $t\colon N \rightarrow \{\tBAS,\tOR,\tAND\}$ such that $t(v) = \tBAS$ if and only if $v$ is a leaf for all $v \in N$.
\end{definition}

The root of $T$ is denoted $R_T$, and the set of children of a node $v$ is denoted $ch(v) = \{w \in N \mid (v,w) \in E\}$. The set of basic attack steps is denoted $\mathrm{BAS}_T = \{v \in N \mid t(v) = \tBAS\}$.

\subsection{Semantics for attack trees}
The semantics of an AT are defined by its successful attacks, i.e., attacks that activate the top node. Formally, an \emph{attack} is a subset $A \subseteq \mathrm{BAS}_T$. For example, in Fig.~\ref{fig:ex_SAT}, $\{p,r\}$ is an attack, corresponding to stealing money by breaking in and then opening the vault. An attack's success is most conveniently expressed by the \emph{structure function}, which is defined recursively as follows:

\begin{definition}
    	\label{def:SAT:sfun}\emph{\cite{lopuhaa2022efficient}}
	Let $T$ be an AT. The \textit{structure function} $f_T\colon N \times 2^{\mathrm{BAS}_T}\to \{0,1\}$
	of $T$ is defined, for a node $v \in N$ and an attack $A \subseteq \mathrm{BAS}_T$, by
\begin{align} 
f_T(v,A) =&
\begin{cases}
    1  & \parbox{55pt}{if~$t(v)=\tOR$} ~\text{and}~%
				\exists u \in ch(v)~\text{s.t} ~f_T(u,A)=1,\\
    1  & \parbox{55pt}{if~$t(v)=\tAND$}~\text{and}~%
				\forall u\in ch(v)~\text{s.t} ~f_T(u,A)=1,\\
    1  & \parbox{55pt}{if~$t(v)=\tBAS$}~\text{and}~%
				v\in A,\\
    0  & \text{otherwise}.
\end{cases}
\end{align}
\end{definition}

An attack $A$ is said to \textit{reach} a node $v$ if $f_T(v,A) = \top$, i.e. it makes $v$ succeed. If no proper subset of $A$ reaches $v$, then $A$ is a \textit{minimal attack on $v$}. The set of minimal attacks on $R_T$ is denoted $\llbracket T \rrbracket$. For example, the AT from Fig.~\ref{fig:ex_SAT}, has three successful attacks:
	$\{r,q\}$, $\{r,p\}$, and $\{r,q,p\}$. The first two are minimal,
	so we have: $\llbracket T \rrbracket=\{\{r,q\},\{r,p\}\}$.

Discussion regarding attacks and semantics for ATs are presented in~\cite{lopuhaa2022efficient}. Note that adding BASes to an attack will not make it less successful; hence the successful attacks are determined by $\llbracket T \rrbracket$. This leads to the following definition of the semantics.

\begin{definition}\label{def:SAT:semantics}
	The \emph{semantics of an AT} $T$ is its suite of minimal attacks $\llbracket T \rrbracket $.
\end{definition}

\subsection{Security metrics for attack trees}\label{sec:crisp_metrics}

Quantitative AT analysis may concern various attributes, such as cost, time, damage, etc. To handle all these attributes in a generic way, analysis algorithms work over a so-called \emph{attribute domain} $(V,\triangledown,\vartriangle)$. Here $V$ is the value domain for the attribute, e.g., $\mathbb{R}_{\geq 0}$ for costs, and $[0,1]$ for probability. Furthermore, $\triangledown$ and $\vartriangle$ are binary operators on $V$, where $\triangledown$ denotes the way values are propagated over an  $\tOR$-gate: If $T = \tOR(a,b)$ and $a,b$ are BASs assigned metric values $x_a,x_b$, then $x_a \triangledown x_b$ is the security value of $T$. Similarly $\vartriangle$ is the operator corresponding to the $\tAND$-gate. For technical reasons we assume $\triangledown$ and $\vartriangle$ satisfy some algebraic properties, which is encoded in the definition of a semiring.

\begin{definition}\emph{\cite{lopuhaa2022efficient}}
A \emph{semiring} is a tuple $(V,\triangledown,\vartriangle)$ where $V$ is a set, $\triangledown$ and $\vartriangle$ are commutative associative binary operators on $V$, and $\vartriangle$ distributed over $\triangledown$ (i.e. $x \vartriangle (y \triangledown z) = (x \vartriangle y) \triangledown (x \vartriangle z)$).
\end{definition}

To assign a metric value to an AT $T$, one chooses a semiring $V$ in which the metric takes value, as well as a BAS value $x_a \in V$ for each BAS $a$; this is encoded as a vector $\vec{x} \in V^{\mathrm{BAS}_T}$. The calculation of $T$ proceeds in two steps: first, we assign values to an attack $A = \{a_1,\ldots,a_n\}$. Since all BASs have to be executed, we set $m_A(\vec{x}) = \bigtriangleup_{i=1}^n x_{a_i}$. This corresponds to the cost/damage/probability/etc. of the attack $A$, given the BAS values $\vec{x}$. Next, we calculate the metric value of $T$ as a whole. To do this, we consider the set of all minimal attacks $\llbracket T \rrbracket  = \{A_1,\ldots,A_m\}$. Since for the top node to be reached one only needs one minimal attack, the metric value for $T$ is calculated via $m_T(\vec{x}) = \bigtriangledown_{i=1}^m m_{A_i}(\vec{x})$.

\begin{example}
	\label{ex:SAT:metric}
	\setlength{\abovedisplayskip}{1.5ex}  
    We consider the \emph{minimal cost} metric that assigns to an AT the minimal cost the attacker needs to spend to successfully reach the top node. This corresponds to the semiring $(\mathbb{N},\min,+)$. Indeed, the cost needed to activate the top node in $\tOR(a,b)$ is the minimum of the costs $x_a$ and $x_b$, as only one of the two children needs to be activated; hence $\triangledown = \min$. Similarly, an $\tAND$-gate needs to activate all children, so their costs need to be added and $\vartriangle = +$. Then given a vector $\vec{x} \in \mathbb{R}_{\geq 0}^{\mathrm{BAS}_T}$ assigning a cost value $x_a \in \mathbb{R}_{\geq 0}$ to each BAS $a$, the metric value of $T$ is defined as $m_T(\vec{x}) = \min_{A \in \llbracket T \rrbracket} \sum_{a \in A} x_a$.  Here $\sum_{a \in A} x_a$ is the total cost of performing an attack $A$, so the metric value corresponds to the cost of the cheapest minimal attack. Consider the AT $T=\mathrm{AND}\big(r,\mathrm{OR}(q,p)\big)$ in Fig.~\ref{fig:ex_SAT}. Recall that $\llbracket T \rrbracket =\{\{r,q\},\{r,p\}\}=\{A_1, A_2\}$, and consider an attribution $\vec{x}$ given by $x_r = 60, x_q = 30, x_p = 5$. Then the metric can be calculated as follows.

	\begin{align*}
	m_{T}(\vec{x}) &= \min\left(\sum_{a \in A_1} x_a,\sum_{a \in A_2} x_a\right) \\
 &= \min(60+30, 60+5) = 65.
	\end{align*}

\end{example}

Formalizing the discussion and example above leads to the following definition.

\begin{definition} \label{def:metric}\emph{\cite{lopuhaa2022efficient}}
Let $T$ be an AT and let $(V,\triangledown,\vartriangle)$ be a semiring.
\begin{enumerate}
    \item An \emph{attribution} of $T$ in $V$ is an element $\vec{x}$ of $V^{\mathrm{BAS}_T}$.
    \item Given an attribution $\vec{x}$, the \emph{metric value} of $T$ given $V$ and $\vec{x}$ is defined as
    \begin{equation}\label{eq:crisp_metric}
        m_T(\vec{x}) = \bigtriangledown_{A\in \llbracket T \rrbracket}  \bigtriangleup_{a \in A} x_a \in V.
    \end{equation}
\end{enumerate}
\end{definition}

As is implicit from the notation, we consider a metric to be a function $m_T\colon V^{\mathrm{BAS}_T}\rightarrow V$ that takes as input the vector $\vec{x}$ of BAS attribute value (e.g. BAS costs), and outputs the AT's security value (e.g. minimal cost needed to succesfully attack the AT). This viewpoint is useful when extending AT metrics to the fuzzy setting in the next section.

\section{Fuzzy metrics for attack trees}\label{sec:fuzzy_metrics}

To define fuzzy AT metrics --- as stated, to the best of our knowledge no such definition exist yet ---  we equip each BAS with a fuzzy element of $V$, i.e., an element of $\mathbf{F}(V)$. Thus, a fuzzy attribution is an element $\vec{\mathsf{x}}$ of $\mathbf{F}(V)^{\mathrm{BAS}_T}$, assigning a fuzzy element $\mathsf{x}_a$ to each BAS $a$. For crisp metrics, the AT's metric value is obtained by applying a function $m_T$ to the crisp attribution vector $\vec{x}$, as outlined in Definition \ref{def:metric}. Analogously, we obtain the fuzzy metric value by applying $\tilde{m}_T$ to $\vec{\mathsf{x}}$, where $\tilde{m}_T$ is the Zadeh extension of $m_T$.

\begin{example}\label{ex:fuzzy_metric_calculate_fastAT}
    Consider the AT $T = \mathrm{AND}(r, \mathrm{OR}(q,p)) $ from Fig.~\ref{fig:ex_SAT}; recall that $ \llbracket T \rrbracket =  \{ \{r,q\}, \{r,p\}\}$. We consider the \emph{minimal time} metric, corresponding to the semiring $(\mathbb{R}_{\geq 0},\min, +)$. For this semiring, consider the fuzzy attribution $\vec{\mathsf{x}} = (\mathsf{x}_r, \mathsf{x}_q, \mathsf{x}_p)$ given by $\mathsf{x}_r = \{50 \mapsto 1, 60 \mapsto 1\}, \mathsf{x}_q = \{0\mapsto 1\}$, and $\mathsf{x}_p = \{5 \mapsto 1\}$, respectively; that is, $q$ and $p$ have crisp time values, and $r$ either takes time $50$ or $60$, with equal possibility.
    
    Since the minimal attacks are $\{r,q\}$ and $\{r,p\}$, the function $m_T\colon V^3 \rightarrow V$ is given by $m_T(x_r,x_q,x_p) = \min(x_r+x_q,x_r+x_p)$ for all $x_r,x_q,x_p \in V$. Then the fuzzy metric value is equal to $\tilde{m}_T(\mathsf{x}_r,\mathsf{x}_q,\mathsf{x}_p)$. Using the definition of Zadeh extension from Definition \ref{def:extension_principle}, the confidence that this fuzzy metric value is equal to a $y \in \mathbb{R}_{\geq 0}$ is equal to
    
    \begin{align*}
        \widetilde{m}_T(\vec{\mathsf{x}})[y] 
                 &=  \sup_{\substack{x_r, x_q, x_p \in \mathbb{R}_{\geq 0}:\\ \min(x_r + x_q, x_r + x_p) = y}} \min \bigl(\mathsf{x}_r[x_r], \mathsf{x}_q[x_q], \mathsf{x}_p[x_p] \bigr).
\end{align*}

Since $\mathsf{x}_q[x_q] \neq 0$ only for $x_q = 0$, where $\mathsf{x}_q[x_q] = 1$, we only need to consider $x_q = 0$, and, for the same reason, we only need to consider $x_p = 5$. Thus the expression above is equal to

\begin{align*}                 
\sup_{\substack{x_r:\\ \min(x_r, x_r + 5) = y}}\min \bigl(\mathsf{x}_r[x_r], 1, 1 \bigr)
                 &=\begin{cases}
                        1, & \textrm{ if $y=50$ or $y=60$},\\
                        0, & \textrm{ otherwise}.
                    \end{cases}
    \end{align*}
so $\widetilde{m}_T(\vec{\mathsf{x}}) = \{50 \mapsto 1, 60 \mapsto 1\}$.

\end{example}

Formally fuzzy AT metrics are then defined as follows.

\begin{definition} \label{def:fuzzy_metric}
Let $T$ be an AT and let $(V,\triangledown,\vartriangle)$ be a semiring.
\begin{enumerate}
    \item A \emph{fuzzy attribution} is an element $\vec{\mathsf{x}}$ of $\mathbf{F}(V)^{\mathrm{BAS}_T}$.

\item Given a fuzzy attribution $\vec{\mathsf{x}}$, the \emph{fuzzy metric value} of $T$ given $V$ and $\vec{\mathsf{x}}$ is defined as $\widetilde{m}_T(\vec{\mathsf{x}})$, where $\widetilde{m}_T\colon \mathbf{F}(V)^{\mathrm{BAS}_T} \rightarrow \mathbf{F}(V)$ is the Zadeh extension of the function $m_T$ from Definition \ref{def:metric}.
    
\end{enumerate}
\end{definition}

More concretely, $\widetilde{m}_T(\vec{\mathsf{x}})$ is the fuzzy element of $V$ defined, for $y \in V$, by
\begin{align}\label{eq:fuzzy_metrics}
\widetilde{m}_T(\vec{\mathsf{x}})[y] &= \sup_{\substack{\vec{x} \in V^{\mathrm{BAS}_T}\colon\\ m_T(\vec{x}) = y}} \min_{v \in \mathrm{BAS}_T} \mathsf{x}_v[x_v]  \nonumber \\
 &= \sup_{\substack{\vec{x} \in V^{\mathrm{BAS}_T}\colon\\ \bigtriangledown_{A\in \llbracket T \rrbracket}  \bigtriangleup_{a \in A} x_a = y}} \min_{v \in \mathrm{BAS}_T} \mathsf{x}_v[x_v]. 
\end{align}

Our choice of using Zadeh's extension to extend crisp AT metrics to fuzzy AT metrics is justified by the fact that Zadeh extension treats the input fuzzy numbers $\mathsf{x}_1,\ldots,\mathsf{x}_n$ as \emph{independent}, i.e., it assumes that there is no nontrivial joint fuzzy distribution on the product space $\prod_i X_i$ of which the $\mathsf{x}_i$ are the marginal distributions \cite{reche2020construction}. This is a standard assumption on BASes (See \cite{pandey2005fault} for a similar viewpoint on fault trees) which we follow. In theory, one could extend the definition to allow non-independent BASes with more complicated joint fuzzy distributions. However, the prevailing viewpoint is that such relations should be explicitly modeled into the AT itself. For example, if the non-independence is due to a common cause affecting the joint distribution of multiple BAS attribute values, then this common cause should be explicitly modeled into the AT framework by replacing the BAS by sub-ATs with shared nodes \cite{pandey2005fault}. We will follow this philosophy and use the Zadeh extension as the natural way to define fuzzy AT metrics.

An alternative way of defining fuzzy AT metrics would be to replace the crisp operators $\triangledown, \vartriangle$ in \eqref{eq:crisp_metric} with their fuzzy counterparts $\widetilde{\triangledown}, \widetilde{\vartriangle}$. However, this does not coincide with our definition, as the following result shows:

\begin{theorem} \label{thm:altmetric}
In general,
\begin{equation} \label{eq:ineq}
\widetilde{m}_T(\vec{\mathsf{x}})  \neq 
\underset{A\in \llbracket T \rrbracket}{\widetilde{\bigtriangledown}}  \underset{a\in A}{\widetilde{\bigtriangleup}} \mathsf{x}_a,
\end{equation}
\end{theorem}

This result is shown by the following example.

\begin{example} \label{ex:ineq}
We continue Example \ref{ex:fuzzy_metric_calculate_fastAT}, where $\widetilde{m}_T(\mathsf{x}_p,\mathsf{x}_q,\mathsf{x}_r) = \{50 \mapsto 1, 60 \mapsto 1\}$. On the other hand,
\begin{align*}
    \underset{A\in \llbracket T \rrbracket}{\widetilde{\bigtriangledown}}  \underset{v\in A}{\widetilde{\bigtriangleup}} \mathsf{x}_v &= \widetilde{\min}\bigl(\mathsf{x}_r \widetilde{+} \mathsf{x}_q, \mathsf{x}_r \widetilde{+} \mathsf{x}_p\bigr).
\end{align*}

One could calculate this fuzzy number in a manner analogous to Example \ref{ex:fuzzy_metric_calculate_fastAT}, but here we show another method that is often more convenient. For a fuzzy number $\mathsf{x} \in \mathbf{F}(\mathbb{R}_{\geq 0})$, define $\mathsf{x}^{(1)} = \{x \in \mathbb{R}_{\geq 0} \mid \mathsf{x}[x] = 1\}$; this is the level 1 $\alpha-$\emph{cut} of $\mathsf{x}$ \cite{Jezewski2017theory}. Then from Definition \ref{def:extension_principle} one can deduce that for $\mathsf{x},\mathsf{y} \in \mathbf{F}(\mathbb{R}_{\geq 0})$ and $f\colon \mathbb{R}_{\geq 0}^2 \rightarrow \mathbb{R}_{\geq 0}$ one has
\[
(\tilde{f}(\mathsf{x},\mathsf{y}))^{(1)} = \{f(x,y) \mid x \in \mathsf{x}^{(1)}, y \in \mathsf{y}^{(1)}\}.
\]

For brevity we abbreviate the right hand side of this equation to $f(\mathsf{x}^{(1)},\mathsf{y}^{(1)})$. It follows that

\begin{align*}
\left(\widetilde{\min}\bigl(\mathsf{x}_r \widetilde{+} \mathsf{x}_q, \mathsf{x}_r \widetilde{+} \mathsf{x}_p\bigr)\right)^{(1)} &= \min((\mathsf{x}_r \widetilde{+} \mathsf{x}_q)^{(1)}, (\mathsf{x}_r \widetilde{+} \mathsf{x}_p)^{(1)})\\
&= \min(\mathsf{x}_r^{(1)} + \mathsf{x}_q^{(1)}, \mathsf{x}_r^{(1)} + \mathsf{x}_p^{(1)})\\
&= \min(\{50,60\} + \{0\}, \{50,60\} + \{5\}) \\
&= \min(\{50,60\},\{55,65\})\\
& = \{50,55,60\}.
\end{align*}
Hence $\left(\underset{A\in \llbracket T \rrbracket}{\widetilde{\bigtriangledown}}  \underset{v\in A}{\widetilde{\bigtriangleup}} \mathsf{x}_v\right)[x] = 1$ if and only if $x \in \{50,55,60\}$. Since this fuzzy number only takes possibility values $0$ and $1$, it follows that 
\[
\underset{A\in \llbracket T \rrbracket}{\widetilde{\bigtriangledown}}  \underset{v\in A}{\widetilde{\bigtriangleup}} \mathsf{x}_v = \{50 \mapsto 1, 55 \mapsto 1, 60 \mapsto 1\} \neq \{50 \mapsto 1, 60 \mapsto 1\} = \widetilde{m}_T(\mathsf{x}_p,\mathsf{x}_q,\mathsf{x}_r).
\]
The `extra' possibility $55 \mapsto 1$ on the LHS comes from comparing the attack $\{r,q\}$ with cost $60+0$ to the attack $\{r,p\}$ with cost $50+5$. In other words, in this comparison $r$ is considered to have costs 50 and 60 simultaneously. By contrast, in the calculation of $\tilde{m}_T(\vec{\mathsf{x}})$ the cost $x_r$ can only have one value at a time.
\end{example}

Equation \eqref{eq:ineq} shows that a priori, there are two ways one can define fuzzy AT metrics. We choose to use the definition of $\widetilde{m}_T(\vec{\mathsf{x}})$ via Zadeh's extension as in Definition \ref{def:fuzzy_metric} for two reasons: first, this accurately captures the independence of the BASes as outlined below Definition \ref{def:fuzzy_metric}. Second, we show in Theorem \ref{theorem:modular_computation} that this definition satisfies modular decomposition, a fundamental property of AT metrics. The RHS of \eqref{eq:ineq} does \emph{not} satisfy modular decomposition, giving another argument why Definition \ref{def:fuzzy_metric} is the preferred definition (see Remark \ref{rem:nomod} below).

\begin{figure}[!t]
     \centering
     \begin{subfigure}[b]{0.35\textwidth}
         \centering
         \includegraphics[width=\textwidth]{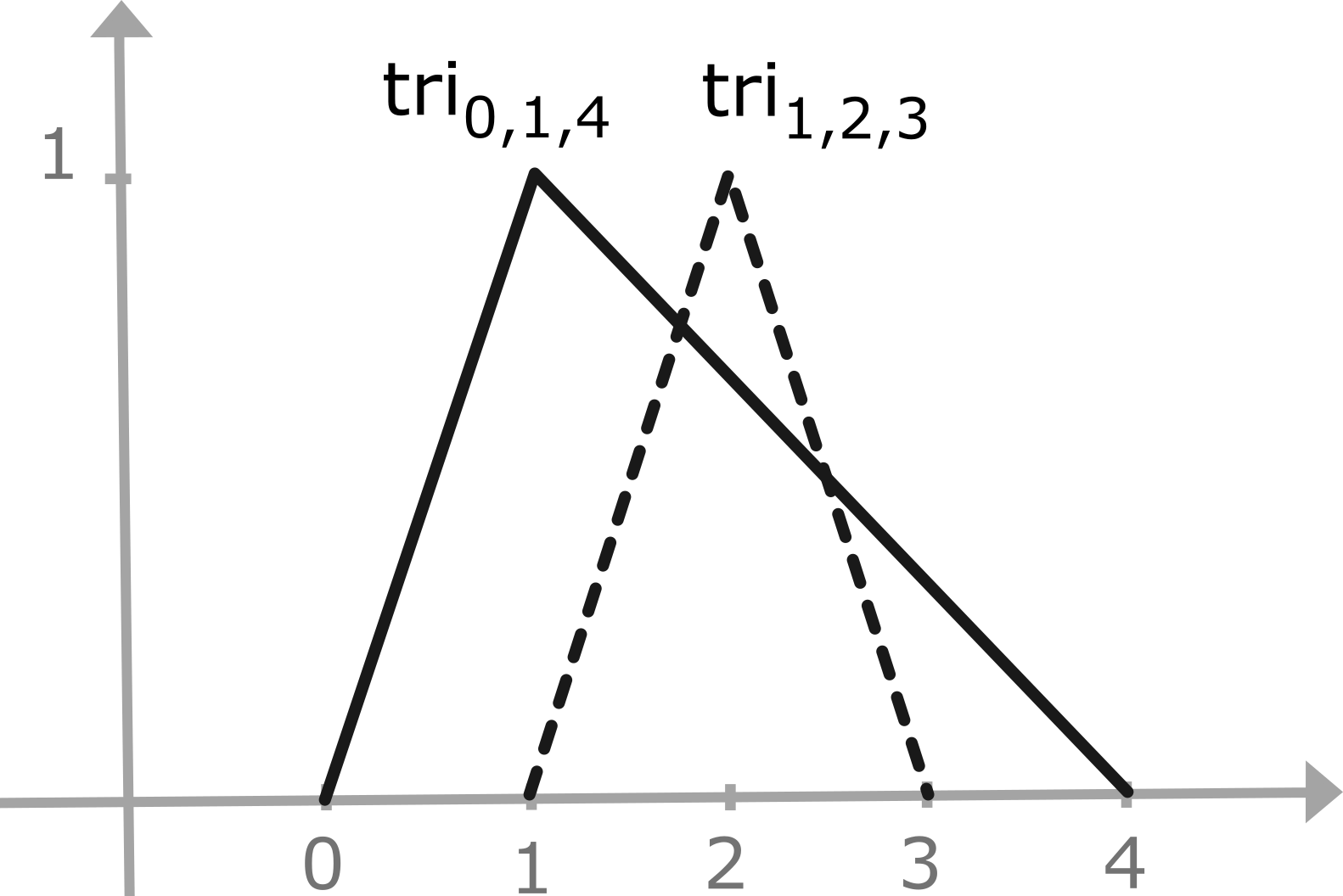}
         \caption{}
         \label{fig:2tri}
     \end{subfigure}
     \hspace{5em}
     \begin{subfigure}[b]{0.35\textwidth}
        \hfill
        \centering         
        \includegraphics[width=\textwidth]{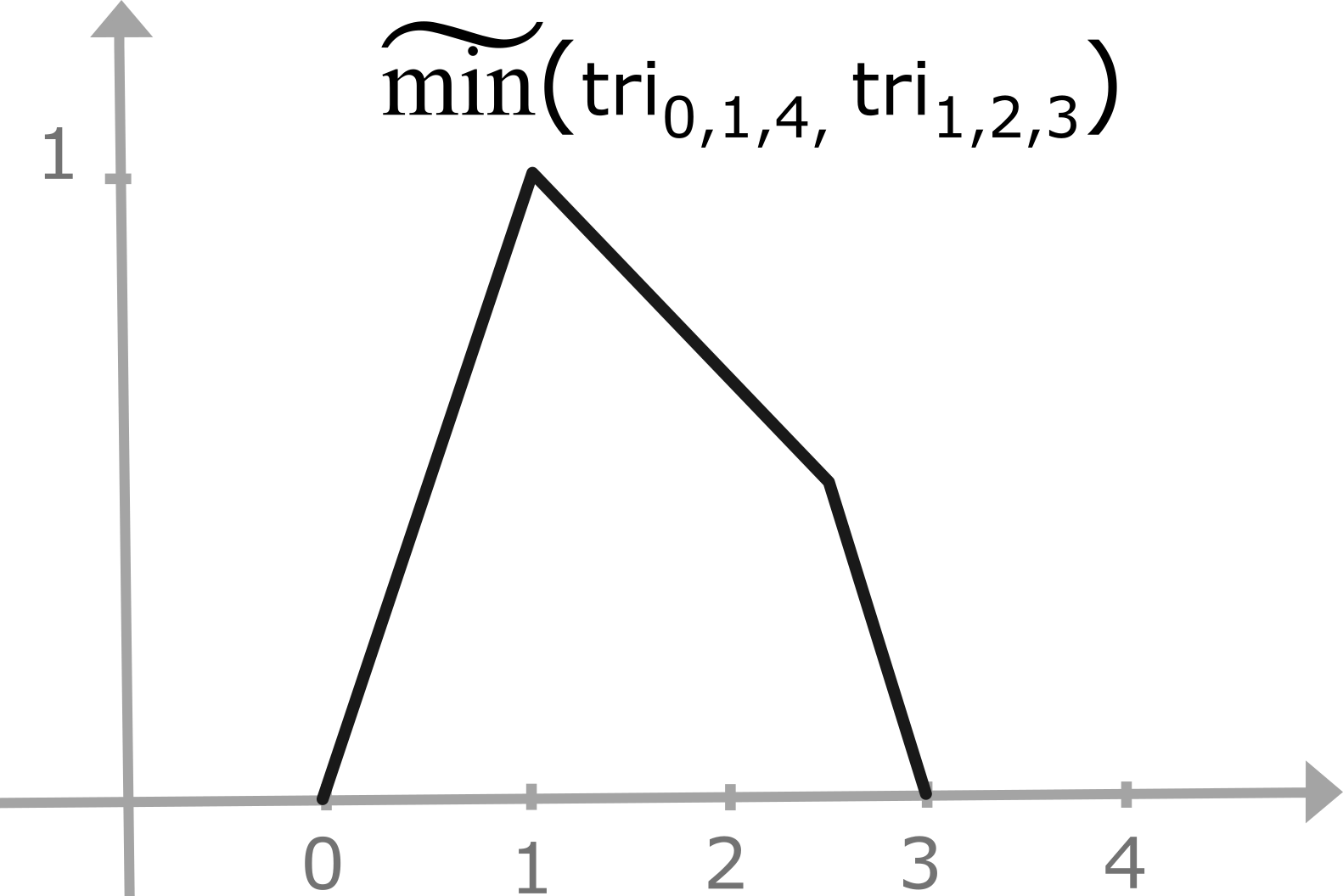}
        \caption{}
        \label{fig:min_2tri}
     \end{subfigure}
     \caption{Two triangular fuzzy numbers and their minimum, as a Zadeh extension of the function $\min$.}
     \label{fig:min}
\end{figure}

\begin{example}
Consider the AT $T = \mathrm{OR}(a,b)$ with the min cost metric, represented by the semiring $(\mathbb{R}_{\geq 0},\min,+)$. As fuzzy attributions consider $\mathsf{x}_a = \mathsf{tri}_{0,1,4}$ and $\mathsf{x}_b = \mathsf{tri}_{1,2,3}$. Then one can show (see Fig.~\ref{fig:min}) that $\widetilde{m}_T(\vec{\mathsf{x}}) = \widetilde{\min}(\mathsf{x}_a,\mathsf{x}_b)$ is given by
\[
\widetilde{\min}(\mathsf{x}_a,\mathsf{x}_b)[x] = \begin{cases}
x, & \textrm{ if $0 \leq x < 1$,}\\
1-\frac{x-1}{3}, & \textrm{ if $1 \leq x < 2.5$,}\\
3-x, & \textrm{ if $2.5 \leq x < 3$,}\\
0, & \textrm{ otherwise}.
\end{cases}
\]
In particular $\widetilde{\min}(\mathsf{x}_a,\mathsf{x}_b)$ is not a triangular fuzzy number. Hence triangular fuzzy numbers are not preserved by the operations inherent to AT analysis. The same holds for other popular subtypes of fuzzy numbers such as rectangular numbers; for this reason, we define fuzzy quantitative AT analysis for general fuzzy numbers in Definition \ref{def:fuzzy_metric}. Finding subtypes of fuzzy numbers that are preserved by AT analysis operations forms an interesting avenue for future research.
\end{example}

\begin{remark}
Besides AT metrics as defined in this paper, in \cite{lopuhaa2022efficient} quantitative analysis for so-called \emph{dynamic ATs} (DATs) is also defined. DATs include a new gate type SAND (``sequential AND'') used when attack steps have to be performed in sequential order; the normal AND-gate allows its children to be performed in parallel. This changes both semantics and quantitative analysis: an attack is now a partially ordered set $(A,\prec)$ rather than just a set $A$ of BASes, to denote the relative timing behaviour of the attack steps; and for quantitative analysis a third binary operation $\triangleright$ is introduced to correspond to SAND-gates, and the metric is defined in terms of these operators.

The results of this paper straightforwardly carry over to the DAT setting. That is, fuzzy DAT metrics are defined as the Zadeh extension of crisp DAT metrics akin to Definition \ref{def:fuzzy_metric}. Furthermore, this definition satisfies modular decomposition, which follows from the modular decomposition of crisp DAT metrics analogous to Theorem \ref{theorem:modular_computation}. As a result, a bottom-up algorithm analogous to Alg.~\ref{alg:BU} calculates fuzzy DAT metrics for treelike DATs.
\end{remark}

\section{Metric computation for ATs}\label{sec:computation_SAT}

To calculate the fuzzy AT metric $\tilde{m}_T(\mathsf{x})$ directly from Definition \ref{def:fuzzy_metric}, one first needs to calculate the function $m_T$, which in return requires one to find $\llbracket T \rrbracket$. In general, this set is of exponential size, making calculation cumbersome for large ATs. Therefore, dedicated algorithms for quantitative AT analysis are needed. For crisp AT metrics these are described in \cite{lopuhaa2022efficient}. In this section, we define a bottom-up algorithm for calculating fuzzy AT metrics for tree-shaped ATs, and we show that its validity follows from the fact that fuzzy AT metrics satisfy modular decomposition. We also show that the BDD-based approach for metric calculation for DAG-shaped ATs from \cite{lopuhaa2022efficient} does not extend to the fuzzy case, and that a radically new approach is needed.

\subsection{Bottom-up algorithm}
The bottom-up algorithm presented in Algorithm~\ref{alg:BU} is adapted from the bottom-up algorithm for crisp AT metrics first presented in~\cite{mauw2006foundations}. It takes as input an AT $T$, a node $v$ of $T$, a semiring $D = (V,\triangledown,\vartriangle)$, and a fuzzy attribution $\vec{\mathsf{x}}$, and outputs a fuzzy value $\widetilde{\mathtt{BU}}(T,v,D,\vec{\mathsf{x}}) \in \mathbf{F}(V)$ assigned to $v$; this value corresponds to the metric value associated to reaching $v$. If $t(v) = \tBAS$, this is simply $\mathsf{x}_v$. If $t(v) = \tOR$, then $\widetilde{\mathtt{BU}}(T,v,D,\vec{\mathsf{x}})$ is obtained by applying $\widetilde{\triangledown}$ to the values associated to the children of $v$; for $t(v) = \tAND$ we instead use $\widetilde{\vartriangle}$. The AT's fuzzy metric value is then given by $\widetilde{\mathtt{BU}}(T,R_T,D,\vec{\mathsf{x}})$.

\SetKwComment{Comment}{/* }{ */}
\begin{algorithm}
\caption{$\widetilde{\mathtt{BU}}$ for tree-structured AT $T$.}\label{alg:BU}
\KwIn{attack tree $T=(N,E,t),$\\
\hspace{1cm} node $v\in N,$\\
\hspace{1cm} semiring attribute domain $D=(V,\triangledown,\vartriangle)$, \\
\hspace{1cm} fuzzy attribution $\vec{\mathsf{x}} \in \mathbf{F}(V)^{\mathrm{BAS}_T}$.}
\KwOut{Fuzzy element $\widetilde{\mathtt{BU}}(T,v,D,\vec{x}) \in \mathbf{F}(V)$.}

\SetAlgoLined
\uIf{$t(v)= \mathtt{OR}$}{
\Return $ \underset{w\in ch(v)}{\widetilde{\bigtriangledown}} \widetilde{\mathtt{BU}}(T,w,D,\vec{\mathsf{x}})$
}
\uElseIf{$t(v)= \mathtt{AND}$}{
\Return $ \underset{w\in ch(v)}{\widetilde{\bigtriangleup}} \widetilde{\mathtt{BU}}(T,w,D,\vec{\mathsf{x}})$
}
\Else (\tcc*[f]{$t(v)=\mathtt{BAS}$}){
{\Return $\mathsf{x}_v$ }
}
\end{algorithm}

\begin{theorem}\label{theorem:BU_SAT}
Let $T$ be a static AT with tree structure, $D = (V,\triangledown,\vartriangle)$ a semiring, and $\vec{\mathsf{x}}$ a fuzzy attribution with values in $V$. Then $\widetilde{m}_T(\vec{\mathsf{x}}) = \widetilde{\mathtt{BU}}(T, R_T , D,\vec{\mathsf{x}})$.
\end{theorem}

\begin{example}
    We apply the algorithm to Example~\ref{ex:fuzzy_metric_calculate_fastAT}. Then the algorithm calculates the metric as follows
    \begin{align*}
        \widetilde{\mathtt{BU}}(T, R_T , D,\vec{\mathsf{x}}) &= \widetilde{\mathtt{BU}}(T,r,D,\vec{\mathsf{x}}) \ \widetilde{\vartriangle} \ \widetilde{\mathtt{BU}}(T,\min(q,p),D,\vec{\mathsf{x}}) \\
        &= \widetilde{\mathtt{BU}}(T,r,D,\vec{\mathsf{x}}) \ \widetilde{\vartriangle} \ \Bigl( \widetilde{\mathtt{BU}}(T,q,D,\vec{\mathsf{x}}) \ \widetilde{\triangledown} \ \widetilde{\mathtt{BU}}(T,p,D,\vec{\mathsf{x}}) \Big)\\
        &= \sup_{\substack{x_r, x_{q \triangledown p} \in \mathbb{R}_{\geq 0}:\\ x_r + x_{q \triangledown p} = y}} \min \Bigl(\mathsf{x}_r[x_r], \sup_{\substack{ x_q, x_p \in \mathbb{R}_{\geq 0}:\\ \min( x_q, x_p) = x_{q \triangledown p} }} \min \bigl( \mathsf{x}_q[x_q], \mathsf{x}_p[x_p] \bigr) \Bigr) \\
        &= \sup_{\substack{x_r, x_q, x_p \in \mathbb{R}_{\geq 0}:\\ x_r + \min( x_q, x_p) = y}} \min \Bigl(\mathsf{x}_r[x_r], \mathsf{x}_q[x_q], \mathsf{x}_p[x_p] \big) \\
        &= \sup_{\substack{ x_r \in \mathbb{R}_{\geq 0}:\\ x_r + \min( 0, 5) = y}} \min \bigl(\mathsf{x}_r[x_r], 1, 1 \big) \\
        &=\begin{cases}
                        1, & \textrm{ if $y=50$ or $y=60$},\\
                        0, & \textrm{ otherwise}.
                    \end{cases} \\
        &=  \{50 \mapsto 1, 60 \mapsto 1\}.
    \end{align*}
\end{example}

The algorithm is efficient as we can see that it is linear in $|E|$, making it vastly more efficient than first calculating $m_T$ and then Zadeh-extending it. The algorithm is generic as it is applicable to popular quantitative metrics in ATs such as cost, damage, skill, probability, etc. \cite{lopuhaa2022efficient}. We should note, however, that the linearity of the time complexity assumes that the fuzzy operations $\widetilde{\triangledown}$ and $\widetilde{\vartriangle}$ take constant time. 

While the algorithm applies only to tree-structured ATs, this covers a large portion of the ATs found in the literature \cite{mauw2006foundations}. As such, the algorithm can be used in many applications.

As we show in the appendix, the proof of Theorem \ref{theorem:BU_SAT} depends on a fundamental property of AT metrics called \emph{modular decomposition}. In the next section, we will explain this and show that fuzzy metrics satisfy this property.

\subsection{Modular decomposition}\label{sec:modular_decomposition}

Modular decomposition is a fundamental property of AT metrics as it facilitates the recursive solution of many problems, which typically improves performance.

For a node $v$ in an AT $T$, let $T_v$ be the AT consisting of all \emph{descendants} of $v$, i.e., the nodes $w$ for which there exists a path $v \rightarrow w$. This is a rooted DAG with root $v$. A \emph{module} is a node $v$ for which $T_v$ is only minimally connected to the rest of $T$:

\begin{definition}\label{def:module}
Let $v \in N \setminus \mathrm{BAS}$. We call node $v$ a \emph{module} if $v$ is the only node in $T_v$ with connections to $T\setminus T_v$.
\end{definition}

For instance, in Fig.~\ref{fig:ex_SAT}, the modules are ``enter the bank'' and ``get money''. Finding the modules of an AT aids in calculating metrics as follows. Given a module $v$, one can split up $T$ into two parts: the sub-AT $T_v$ with root $v$, and the `quotient' $T^v$ obtained by replacing the entire sub-AT $v$ with a single new node, which we will still call $v$ (see Fig.~\ref{fig:modular_metric}). Then one can calculate the metric for $T_v$ to find $\widetilde{m}_{T_v}(\vec{\mathsf{x}})$, and use this as a BAS attribute value for $v$ in $T^v$. One then calculates the metric value for $T^v$ with this new BAS value. In \cite[Thm.~9.2]{lopuhaa2022efficient} it is shown that for crisp metrics this results in the same metric value for $T$ as when one considers the entirety of $T$ at once. As a result, we can split up metric calculations via a divide-and-conquer approach once one has identified the modules. The following theorem shows that this also holds for fuzzy AT metrics.

\begin{theorem}\label{theorem:modular_computation}
Let $(V, \triangledown, \vartriangle)$ be a semiring. Let $v$ be a module in an AT $T$, $\vec{\mathsf{x}} \in \mathbf{F}(V)^{\mathrm{BAS}_T}$ be a fuzzy attribution for $T$. Let $\vec{\mathsf{x}}_v \in \mathbf{F}(V)^{\mathrm{BAS}_{T_v}}$ be the fuzzy attribution for $T_v$ obtained from restricting $\mathsf{x}$, i.e., $(\vec{\mathsf{x}}_v)_w = \mathsf{x}_w$ for all $w \in \mathrm{BAS}_{T_v}$. Let $T^v$ be the AT obtained by replacing $T_v$ in $T$ by a single BAS still called $v$. Let $\vec{\mathsf{x}}^v \in \mathbf{F}(V)^{\mathrm{BAS}_{T^v}}$ be a fuzzy attribution for $T^v$ given by

\begin{align*}
\mathsf{x}^v_{v'} &= 
\begin{cases}
    \mathsf{x}_{v'} , & v'\neq v,\\
    \widetilde{m}_{T_v}(\vec{\mathsf{x}}), & v'= v.
\end{cases}
\end{align*}
Then $\widetilde{m}_{T}(\vec{\mathsf{x}}) = \widetilde{m}_{T^v}(\vec{\mathsf{x}}^v)$.
\end{theorem}

The theorem is the extension of Theorem 9.2 of~\cite{lopuhaa2022efficient}. The proof of Theorem~\ref{theorem:modular_computation} is shown in the appendix. In a treelike AT, every node is a module, and applying modular decomposition then yields Theorem \ref{theorem:BU_SAT}.

\begin{remark} \label{rem:nomod}
In the same way that Theorem \ref{theorem:modular_computation} can be used to prove Theorem \ref{theorem:BU_SAT}, it can also be used to show that the alternative definition of fuzzy AT metrics in the RHS of \eqref{eq:ineq} does \emph{not} satisfy modular decomposition. Namely, if the alternative definition would satisfy modular decomposition, Alg.~\ref{alg:BU} would also calculate the alternative definition for treelike ATs. However, since this does not conform to our Definition \ref{def:fuzzy_metric} even for treelike ATs (see Theorem \ref{thm:altmetric}), we conclude that the alternative definition does not satisfy modular decomposition.
\end{remark}

\begin{figure}[t]
    
    \begin{center}
    \includegraphics[width=0.6\textwidth]{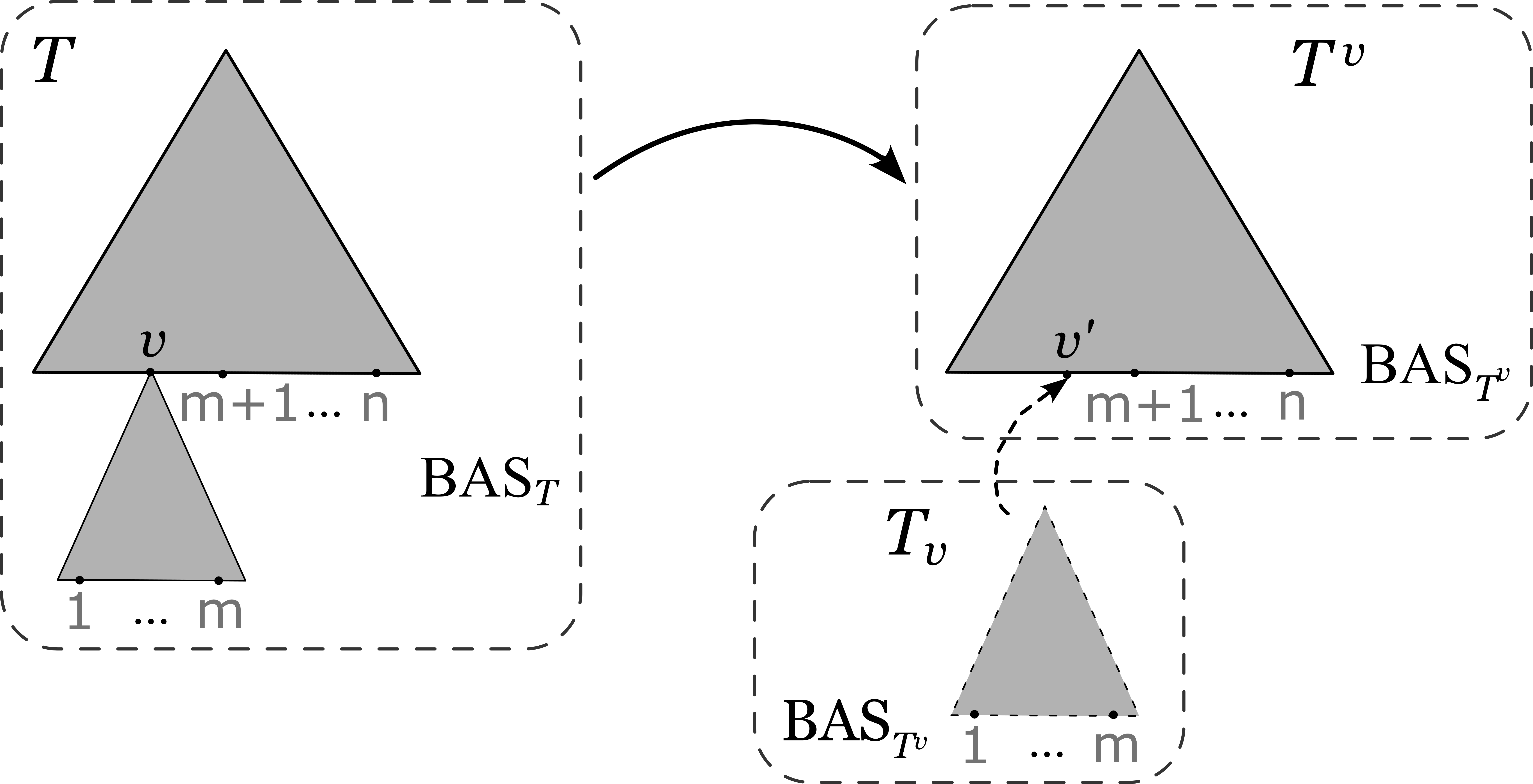}
    \end{center}
    \caption{Calculation of $\widetilde{m}_T(\vec{\mathsf{x}})$ can be done by computing $\widetilde{m}_{T^v}(\vec{\mathsf{x}}^v)$, where $v' \in \text{BAS}_{T^v}$ is assigned with fuzzy attribute $\widetilde{m}_{T_v}(\vec{\mathsf{x}}_v)$.}
    \label{fig:modular_metric}
\end{figure}

\subsection{Computations for DAG ATs}
Directed acyclic graph (DAG) ATs refer to ATs in which a node has more than one parent~\cite{lopuhaa2022efficient}. Fig.~\ref{fig:DAG} visualizes an AT with DAG structure. Unfortunately, Alg.~\ref{alg:BU}, does not correctly compute the (fuzzy) metric value of DAG-shaped ATs. The reason for this is that the algorithm does not detect whether a node's child is shared with another node or not, which leads to double counting of a child's metric value.
\begin{example}
 Let $\mathsf{x}_u=\{1 \mapsto 1\}, \mathsf{x}_v=\{0 \mapsto 1, 3 \mapsto 1\}, \mathsf{x}_w=\{1 \mapsto 1\},$ and $ D=\{\mathbb{N}, \min, + \}$. The min cost computation for the DAG AT shown in Fig.~\ref{fig:DAG} using algorithm~\ref{alg:BU} gives $\widetilde{\mathtt{BU}}(T, R_T , \mathsf{x}, D) = \widetilde{\min} (\mathsf{x}_u, \mathsf{x}_v) \ \widetilde{+} \ \widetilde{\min} (\mathsf{x}_v, \mathsf{x}_w) =  \{0\mapsto 1, 1 \mapsto 1\} \ \widetilde{+} \ \{0\mapsto 1, 1 \mapsto 1\} = \{0 \mapsto 1, 1 \mapsto 1, 2 \mapsto 1 \}$, whereas $\widetilde{m}_T(\mathsf{x}_u, \mathsf{x}_v, \mathsf{x}_w) = \{0 \mapsto 1, 2 \mapsto 1\}$. 
\end{example}

\begin{figure}[!t]
     \centering
     \begin{subfigure}[b]{0.2\textwidth}
         \centering
         \includegraphics[width=\textwidth]{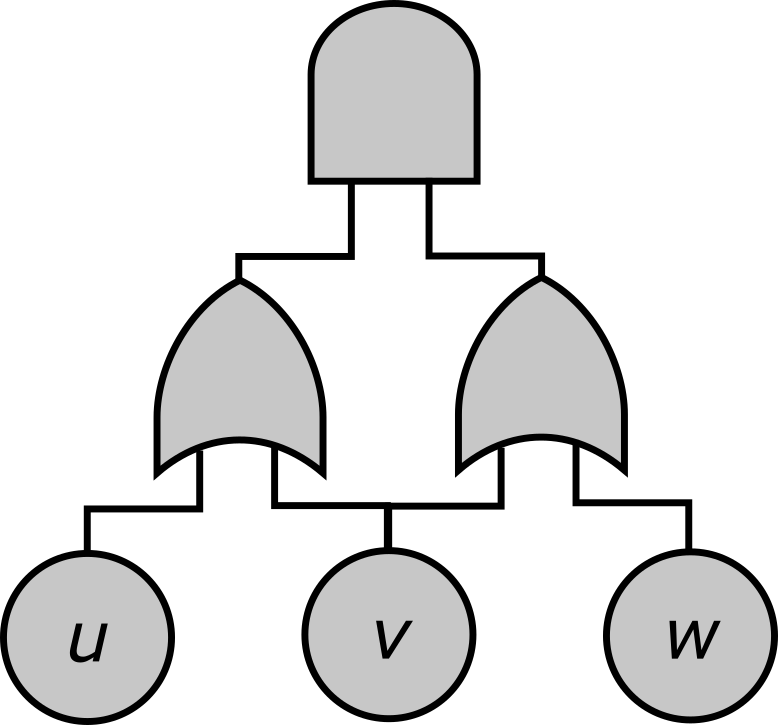}
         \caption{}
         \label{fig:DAG}
     \end{subfigure}
     \hspace{5em}
     \begin{subfigure}[b]{0.15\textwidth}
        \hfill
        \centering         
        \includegraphics[width=\textwidth]{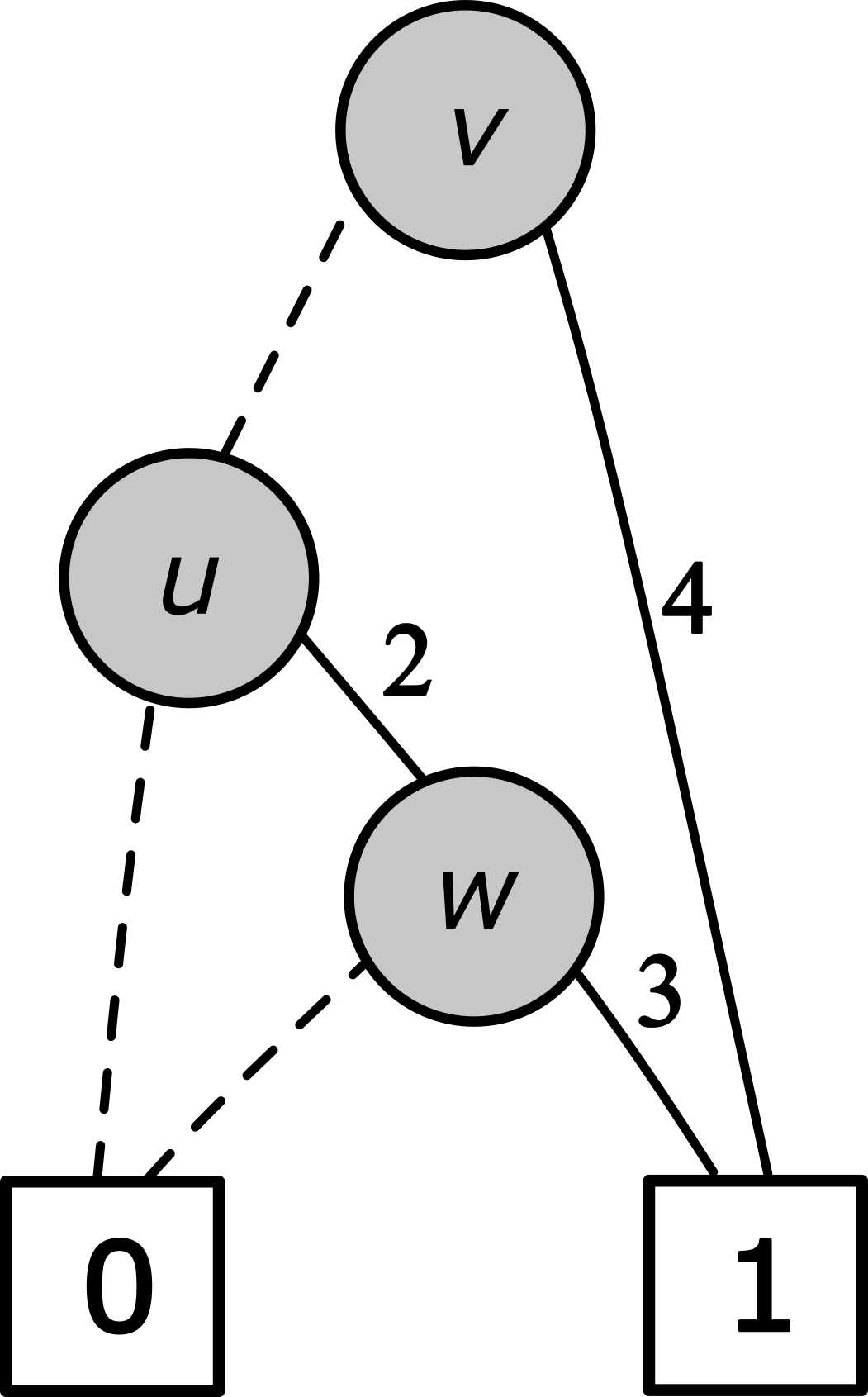}
        \caption{}
        \label{fig:BDD_DAG}
     \end{subfigure}
     \caption{A DAG AT (a), and its BDD (b).}
     \label{fig:ex_DAG}
\end{figure}

For crisp metrics, this was solved by the BDD-based approach introduced in~\cite{lopuhaa2022efficient}. Boolean functions are compactly represented by a binary decision diagram(BDD), a type of directed acyclic graph. One can apply this to the structure function of an AT as in Fig. \ref{fig:BDD_DAG}: as one can see, each nonleaf is labeled with a BAS and has two outgoing edges, while the leafs are labeled $\mathtt{0}$ and $\mathtt{1}$. For a given attack $A$, the BDD evaluates $f_T(R_T,A)$ as follows: at a node with label $v$, follow the dashed line if $v \notin A$, and the nondashed line if $v \in A$. The leaf in which one ends up holds the value of $f_T(R_T,A)$. Every Boolean function can be represented as a BDD, and although the corresponding BDD is worst-case of exponential size, BDDs are usually quite compact.

The BDD can also be used to calculate (crisp) AT metrics. We showcase this for the minimal cost metric, but it can be applied to other metrics, so long as the corresponding semiring is \emph{absorbing} (see \cite{lopuhaa2022efficient}). Minimal cost is calculated as follows: for each BAS $v$, the cost $x_v$ is attached to the nondashed edges originating from BDD nodes with label $v$, while each dashed edge gets label $0$ (see Fig. \ref{fig:BDD_DAG}). Then the attack with minimal cost corresponds to the shortest path from $R_T$ to $\mathtt{1}$ in the BDD; since the BDD is acyclic this computation is linear in the size of the BDD. In total, this means that this is worst-case exponential in the size of the AT, but in practice the calculation is quite fast.

Unfortunately, this approach no longer works for fuzzy AT metrics. The reason is that this approach assumes that the metric arises from a semiring, in particular, that distributivity holds. As the following example shows, if $(V,\triangledown,\vartriangle)$ is a semiring, then $(\mathbf{F}(V),\widetilde{\triangledown},\widetilde{\vartriangle})$ is no longer a semiring, because distributivity no longer holds. It is therefore no surprise that the BDD method no longer works either.

\begin{example}
Let $(V,\triangledown,\vartriangle) = (\mathbb{R}_{\geq 0}, \min,+)$, and consider the fuzzy elements $\mathsf{x} = \{0 \mapsto 1, 2 \mapsto 1\}$ and $\mathsf{y} = \mathsf{z} = \{0 \mapsto 1\}$. Then using the methods from Example \ref{ex:ineq}, we find that
\begin{align*}
\widetilde{\min}(\mathsf{x}\widetilde{+}\mathsf{y},\mathsf{x}\widetilde{+}\mathsf{z}) &= \widetilde{\min}(\{0 \mapsto 1, 2 \mapsto 1\},\{0 \mapsto 1, 2 \mapsto 1\}) \\
&= \{0 \mapsto 1, 1\mapsto 1, 2 \mapsto 1\},\\
\mathsf{x}\widetilde{+}\widetilde{\min}(\mathsf{y},\mathsf{z}) &= \{0 \mapsto 1, 2 \mapsto 1\} \widetilde{+} \{0 \mapsto 1\} \\
&= \{0 \mapsto 1, 2 \mapsto 1\}.
\end{align*}
Hence $(\mathbf{F}(\mathbb{R}_{\geq 0}),\widetilde{\min},\widetilde{+})$ is not distributive, and in particular not a semiring.
\end{example}

The reason that distributivity fails for fuzzy numbers is that, as we discussed in Section \ref{sec:fuzzy_metrics}, a Zadeh-extended operator like $\widetilde{+}$ acts as though its two arguments are independent. However, in an expression like $\widetilde{\min}(\mathsf{x}\widetilde{+}\mathsf{y},\mathsf{x}\widetilde{+}\mathsf{z})$ the arguments $\mathsf{x}\widetilde{+}\mathsf{y}$ and $\mathsf{x}\widetilde{+}\mathsf{z}$ are typically not independent. This ensures that distributivity is not retained under Zadeh extension.

Since the BDD method used for crisp AT metrics does not work, a new method is needed for calculating fuzzy metrics for DAG-like ATs. This is beyond the scope of this paper. One possible way to approach this problem is to find a way to keep track of the `double counting' that occurs when applying $\widetilde{\mathtt{BU}}$ to DAG-like ATs, and eliminate it at the end of the algorithm. Such an approach would require a radically new, strategy, and we therefore leave it to future work.

\section{Conclusion and future work}\label{sec:conclude}
In this paper we define a mathematical formulation for deriving AT fuzzy metrics values. In our knowledge, fuzzy theory has been applied in FTs for imprecise data, but fuzzy quantitative metrics remain somewhat implicitly defined. The definition we provide is explicit and generic for commonly used quantitative metrics. Moreover, this definition can be used to better capture uncertainty in quantitative metrics values. 
In addition, this paper introduces an efficient algorithm to calculate AT metrics with fuzzy attribution. The proposed algorithm is linear in $|E|$, as opposed to the definition of fuzzy metrics which requires calculation of crisp metrics followed by fuzzy operators. The algorithm works for tree-like structure models that satisfy modular decomposition. 

In the future, we want to develop an algorithm for fuzzy metrics computation on DAG ATs. For that aim, the algorithm should address the non-semiring property of fuzzy operators and the DAG structure on ATs. Another avenue for future research is the development of subtypes of fuzzy numbers that are preserved by (Zadeh-extended) arithmetic operations inherent to AT analysis, such as $\min$ and $\max$. Upon formally defining such subtypes, these can then be used to implement quantitative analysis algorithms efficiently.


\subsubsection*{Acknowledgement} This research has been partially funded   by ERC Consolidator grant 864075 CAESAR and the European Union’s Horizon 2020 research and innovation programme under the Marie Skłodowska-Curie grant agreement No. 101008233.

\subsubsection*{Disclosure of Interests} The authors have no competing interests to declare that are relevant to the content of this article.

\bibliographystyle{splncs04}
\bibliography{main}

\section*{Appendix}\label{sec:appendix}

\subsection*{Proof of Theorem~\ref{theorem:modular_computation}}

\begin{proof}
    We use Fig.~\ref{fig:modular_metric} for illustration. We enumerate the BASes of $T$ as $b_1,\ldots,b_n$, such that the first $b_1,\ldots,b_m$ are those BASes that are in $T_v$. Under this enumeration, we can write $\vec{\mathsf{x}} = (\mathsf{x}_1, \dots, \mathsf{x}_n)$, and the definition of $\vec{\mathsf{x}}^v$ then becomes $\vec{\mathsf{x}}^v = \bigl(\widetilde{m}_{T_v}(\mathsf{x}_1, \dots, \mathsf{x}_m), \mathsf{x}_{m+1},\dots, \mathsf{x}_n \bigr)$. To prove Theorem~\ref{theorem:modular_computation}, we need to prove that
    \begin{equation}\label{eq:decomposition}
        \widetilde{m}_T(\mathsf{x}_1,\ldots,\mathsf{x}_{n}) = \widetilde{m}_{T^v} \bigl( \widetilde{m}_{T^v}(\mathsf{x}_1,\ldots,\mathsf{x}_m),\mathsf{x}_{m+1},\ldots,\mathsf{x}_{n} \bigr).
    \end{equation}
    By Theorem 9.2 in~\cite{lopuhaa2022efficient} modular analysis works for crisp metrics, i.e., for all $\vec{x} \in V^n$ we have
    \begin{equation*}
        m_T(x_1,\ldots,x_{n}) = m_{T^v}(m_{T_v}(x_1,\ldots,x_m),x_{m+1},\ldots,x_{n}).
    \end{equation*}
    Applying Theorem~\ref{theorem:modular_decomposition} below with $f=m_{T^v}$, $g=m_{T_v}$, and $h=m_T$, we get~\eqref{eq:decomposition}. Thus, the theorem is proven.
\end{proof}


\begin{theorem}\label{theorem:modular_decomposition}
Let $m < n \in \mathbb{Z}_{\geq 1}$, and $X$ be a set. Let $f\colon X^{n-m+1} \rightarrow X$ and $g\colon X^m \rightarrow X$ be functions. Let $h\colon X^{n}\rightarrow X$ be defined by
\begin{equation*}
    h(x_1,\ldots,x_{n}) = f(g(x_1,\ldots,x_m),x_{m+1},\ldots,x_{n}) \quad \text{for all } x_i \in X. 
\end{equation*}

Then for all $\mathsf{x}_i \in \mathbf{F}(X)$ it holds that

\begin{equation*}
\tilde{h}(\mathsf{x}_1,\ldots,\mathsf{x}_{n}) = \tilde{f}(\tilde{g}(\mathsf{x}_1,\ldots,\mathsf{x}_m),\mathsf{x}_{m+1},\ldots,\mathsf{x}_{n}).
\end{equation*}

\end{theorem}

\begin{proof}
    \begin{align*}
        &\Tilde{f}(\Tilde{g}(\mathsf{x}_1, \dots, \mathsf{x}_m),\mathsf{x}_{m+1}, \dots, \mathsf{x}_{n})[x] \\
        &= \underset{x = f(
        x_g, x_{m+1},\dots, x_{n})}{\sup}
  \min\limits_{}\ \biggl(  \Bigl(  \underset{ x_{g} = g(x_1, \dots, x_m)}{\sup}
  \min\limits_{}\ \bigl( \mathsf{x}_{1}{[x_1]},\dots, \mathsf{x}_m[x_m]  \bigr) \Bigr)  ,\mathsf{x}_{m+1}[x_{m+1}],\dots, \mathsf{x}_{n}[x_{n}]  \biggr) \\
        &= \sup_{\substack{x = f(x_g,x_{m+1},\ldots,x_{n}),\\
  x_g = g(x_1,\ldots,x_m)}} \min \Bigl( \mathsf{x}_1[x_1],\ldots,\mathsf{x}_{n}[x_{n}] \Bigr) \\
        &= \underset{x = h(x_1,\dots, x_{n})}{\sup}
  \min\limits_{}\ \Bigl( \mathsf{x}_{1}[x_1],\dots, \mathsf{x}_{n}[x_{n}]  \Bigr) \\
        &=\Tilde{h}(\mathsf{x}_1, \dots, \mathsf{x}_{n})[x].
    \end{align*}
\end{proof}

\subsection*{Proof of Theorem~\ref{theorem:BU_SAT}}
\begin{proof}
We prove this by induction on the number $n$ of non-leaf nodes of $T$. If $n = 0$, then $T$ consists of a single BAS $b$ and $\widetilde{m}_T(\mathsf{x}_b) = \mathsf{x}_b = \widetilde{\mathtt{BU}}(T, R_T, D, \mathsf{x}_b)$.

If $n = 1$, then either $T = \mathrm{AND}(b_1,\ldots,b_k)$ or $T = \mathrm{OR}(b_1,\ldots,b_k)$ for BASes $b_1,\ldots,b_k$. In the AND-case, on has $\llbracket T \rrbracket = \{\{b_1,\ldots,b_k\}\}$, so $m_T(\vec{x}) = \bigtriangleup_{i=1}^k x_{b_i}$; hence
\begin{align*}
\widetilde{\mathtt{BU}}(T, R_T, D, \vec{\mathsf{x}}) &= \underset{1 \leq i \leq k}{\widetilde{\bigtriangleup}} \widetilde{\mathtt{BU}}(T,b_i,D,\vec{\mathsf{x}}) \\
&= \underset{1 \leq i \leq k}{\widetilde{\bigtriangleup}} \mathsf{x}_{b_i} \\
&= \widetilde{m}_T(\vec{\mathsf{x}}).
\end{align*}
In the OR-case one has $\llbracket T \rrbracket = \{\{b_1\},\ldots,\{b_k\}\}$, so $m_T(\vec{x}) = \bigtriangledown_{i=1}^k x_{b_i}$. The rest of the proof is then analogous to the AND-case. Together, this covers the case $n = 1$.

Now assume the theorem has been proven for all $n' < n$, and let $T$ be a tree-structured AT with $n$ non-leaf nodes. Let $v$ be any non-leaf, non-root node of $T$. Since $T$ is tree-structured, $v$ is a module. By the induction hypothesis we have
\begin{equation} \label{eq:buproof1}
\widetilde{\mathtt{BU}}(T_v, v, D, \vec{\mathsf{x}}) = \widetilde{m}_{T_v}(\mathsf{x}).
\end{equation}
Furthermore, it is straightforward to show by induction that for all nodes $w$ in $T_v$ one has $\widetilde{\mathtt{BU}}(T_v, w, D, \vec{\mathsf{x}}_v) = \widetilde{\mathtt{BU}}(T, w, D, \vec{\mathsf{x}})$. In particular, we have, using~\eqref{eq:buproof1},
\begin{equation} \label{eq:buproof2}
\widetilde{\mathtt{BU}}(T, v, D, \vec{\mathsf{x}}) = \widetilde{\mathtt{BU}}(T_v, v, D, \vec{\mathsf{x}}_v) = \widetilde{m}_{T_v}(\mathsf{x}).
\end{equation}
We now apply the same reasoning on $T^v$, where for all nodes $w$ it holds that $\widetilde{\mathtt{BU}}(T^v, w, D, \vec{\mathsf{x}}^v) = \widetilde{\mathtt{BU}}(T, w, D, \vec{\mathsf{x}})$. Indeed, for BASes $w \neq v$ this is immediately true, for $w=v$ it is shown in~\eqref{eq:buproof2}, and for non-BASes this follows from a straightforward induction proof. As a result, we get
\begin{align}
\widetilde{\mathtt{BU}}(T, R_T, D, \vec{\mathsf{x}}) &= \widetilde{\mathtt{BU}}(T^v, R_{T^v}, D, \vec{\mathsf{x}}) \nonumber \\
&= \widetilde{m}_{T^v}(\vec{\mathsf{x}}^v) \label{eq:buproof3}\\
&= \widetilde{m}_T(\vec{\mathsf{x}}); \label{eq:buproof4}
\end{align}
Here \eqref{eq:buproof3} follows from the induction hypothesis applied to $T^v$, and~\eqref{eq:buproof4} follows from Theorem \ref{theorem:modular_computation}.
\end{proof}

\end{document}